\documentclass[runningheads]{llncs}

\usepackage{graphicx}
\usepackage[T1]{fontenc}
\usepackage{graphicx}
\usepackage{cite}
\usepackage{amsmath}
\usepackage{algorithm2e}
\usepackage{pxfonts}
\usepackage{lineno}
\usepackage[in]{fullpage}

\def\calP{\mathcal{P}}

\def\calA{\mathcal{A}}

\def\Ed{\mathsf{E}\mathrm{d}}

\newtheorem{observation}{Observation}

\begin{document}

\title{Computing the Center of Uncertain Points on Cactus Graphs\thanks{A preliminary version of this paper appeared in \textit{Proceeding of the 34th International Workshop on Combinatorial Algorithms (IWOCA 2023)}.}}

\author{Ran Hu\inst{1}\and Divy H. Kanani\inst{2}\and Jingru Zhang\inst{2}}
\authorrunning{R. Hu et al.}

\institute{Rensselaer Polytechnic Institute, Troy, NY 12180, USA\\
\email{hur6@rpi.edu}\\
\and 
Cleveland State University, Cleveland, OH 44115, USA\\
\email{d.kanani@vikes.csuohio.edu, j.zhang40@csuohio.edu}  
}

\maketitle              
\begin{abstract}
In this paper, we consider the (weighted) one-center problem of uncertain points on a cactus graph. Given are a cactus graph $G$ and a set of $n$ uncertain points. Each uncertain point has $m$ possible locations on $G$ with probabilities and a non-negative weight. The (weighted) one-center problem aims to compute a point (the center) $x^*$ on $G$ to minimize the maximum (weighted) expected distance from $x^*$ to all uncertain points. No previous algorithm is known for this problem. In this paper, we propose an $O(|G| + mn\log mn)$-time algorithm for solving it. Since the input is $O(|G|+mn)$, our algorithm is almost optimal.

\end{abstract}

\section{Introduction}\label{introduction}
Problems on uncertain data have attracted an increasing amount of attention due to the observation that many real-world measurements are inherently accompanied with uncertainty. For example, the $k$-center model has been considered a lot on uncertain demands in facility locations~\cite{ref:AbamPr22,ref:QuanLi21,ref:BenkocziAn05,ref:WangCo19,ref:HuangSt17,ref:AverbakhFa05,ref:AverbakhMi97,ref:KeikhaCl21}. Due to the prevalence of tree-like graphs~\cite{ref:BurkardAL98,ref:ZmazekTh04,ref:BaiTh17,ref:Ben-MosheEf07,ref:GranotOn94,ref:BhattacharyaIm14} in facility locations, in this paper, we study the (weighted) one-center problem of uncertain points on a cactus-graph network. 

Let $G = (V,E)$ be a cactus graph where any two cycles do not share edges. Every edge $e$ on $G$ has a positive length. A point $x =(u,v,t)$ on $G$ is characterized by being located at a distance of $t$ on edge $(u,v)$ from vertex $u$. Given any two points $p$ and $q$ on $G$, the distance $d(p,q)$ between $p$ and $q$ is defined as the length of their shortest path on $G$.
   
Let $\calP$ be a set of $n$ uncertain points $P_1, P_2, \cdots, P_n$ on $G$. Each $P_i\in\calP$ has $m$ possible locations (points) $p_{i1}, p_{i2}, \cdots, p_{im}$ on $G$. Each location $p_{ij}$ is associated with a probability $f_{ij}\geq 0$ for $P_i$ appearing at $p_{ij}$. Additionally, each $P_i\in\calP$ has a weight $w_i\geq 0$. 

Assume that all given points (locations) on any edge $e\in G$ are given sorted so that when we visit $e$, all points on $e$ can be traversed in order. 

Consider any point $x$ on $G$. For any $P_i\in\calP$, the (weighted) expected distance $\Ed(P_i,x)$ from $P_i$ to $x$ is defined as $w_i\cdot \sum_{j=1}^{m}f_{ij}d(p_{ij}, x)$. The center of $G$ with respect to $\calP$ is defined to be a point $x^*$ on $G$ that minimizes the maximum expected distance $\max_{1\leq i\leq n}\Ed(P_i,x)$. The goal is to compute center $x^*$ on $G$. 

If $G$ is a tree network, then center $x^*$ can be computed in $O(mn)$ time by~\cite{ref:WangCo17}. To the best of our knowledge, however, no previous work exists for this problem on cacti. In this paper, we propose an $O(|G|+ mn\log mn)$-time algorithm for solving the problem where $|G|$ is the size of $G$. Note that our result matches the $O(|G|+n\log n)$ result~\cite{ref:Ben-MosheEf07} for the weighted {\em deterministic} case where each uncertain point has exactly one location. 

\subsection{Related Work}\label{relatedwork} 
The deterministic one-center problem on graphs have been studied a lot. On a tree, the (weighted) one-center problem has been solved in linear time by Megiddo~\cite{ref:MegiddoLi83}. On a cactus, an $O(|G|+n\log n)$ algorithm was given by Ben-Moshe~\cite{ref:Ben-MosheEf07}. Note that the unweighted cactus version can be solved in linear time~\cite{ref:LanAl99}. When $G$ is a general graph, the center can be found in $O(|E|\cdot |V|\log |V|)$ time~\cite{ref:KarivAn79}, provided that the distance-matrix of $G$ is given. See~\cite{ref:BaiTh17,ref:YenTh12,ref:ZmazekTh04} for variations of the general $k$-center problem.  

When it comes to uncertain points, a few of results for the one-center problem are available. When $G$ is a path network, the center of $\calP$ can be found in $O(mn)$ time~\cite{ref:WangAn16}. On tree graphs, the problem can be addressed in linear time~\cite{ref:WangCo19} as well. See \cite{ref:WangCo19,ref:HuangSt17,ref:KeikhaCl21} for the general $k$-center problem on uncertain points. 

\subsection{Our Approach}\label{ourapproach}
Lemma~\ref{lem:reduction} shows that the general one-center problem can be reduced in linear time to a {\em vertex-constrained} instance where all locations of $\calP$ are at vertices of $G$ and every vertex of $G$ holds at least one location of $\calP$. Our algorithm focuses on solving the vertex-constrained version. 

As shown in~\cite{ref:BurkardAL98}, a cactus graph is indeed a block graph and its skeleton is a tree where each node uniquely represents a cycle block, a graft block (i.e., a maximum connected tree subgraph), or a hinge (a vertex on a cycle of degree at least $3$) on $G$. Since center $x^*$ lies on an edge of a circle or a graft block on $G$, we seek for that block containing $x^*$ by performing a binary search on its tree representation $T$. Our $O(mn\log mn)$ algorithm requires to address the following problems. 

We first solve the one-center problem of uncertain points on a cycle. Since each $\Ed(P_i,x)$ is piece-wise linear but non-convex as $x$ moves along the cycle, our strategy is computing the local center of $\calP$ on every edge. Based on our useful observations, we can resolve this problem in $O(mn\log mn)$ time with the help of the dynamic convex-hull data structure~\cite{ref:AgarwalDy95,ref:BrodalDy02}.  

Two more problems are needed to be addressed during the search for the node containing $x^*$. First, given any hinge node $h$ on $T$, the problem requires to determine if center $x^*$ is on $h$, i.e., at hinge $G_h$ $h$ represents, and otherwise, which split subtree of $h$ on $T$ contains $x^*$, that is, which hanging subgraph of $G_h$ on $G$ contains $x^*$. In addition, a more general problem is the {\em center-detecting} problem: Given any block node $u$ on $T$, the goal is to determine whether $x^*$ is on $u$ (i.e., on block $G_u$ on $G$), and otherwise, which split tree of the $H$-subtree of $u$ on $T$ contains $x^*$, that is, which hanging subgraph of $G_u$ contains $x^*$. 

These two problems are more general problems on cacti than the tree version~\cite{ref:WangCo17} since each $\Ed(P_i, x)$ is no longer a convex function in $x$ on any path of $G$. We however observe that the median of any $P_i\in\calP$ always fall in the hanging subgraph of a block whose probability sum of $P_i$ is at least $0.5$. Based on this, with the assistance of other useful observations and lemmas, we can efficiently solve each above problem in $O(mn)$ time.  

\textbf{Outline.} In Section~\ref{sec:pre}, we introduce some notations and observations. In Section~\ref{sec:algforcycle}, we present our algorithm for the one-center problem on a cycle. In Section~\ref{sec:algorithm}, we discuss our algorithm for the problem on a cactus. In Section~\ref{sec:reduction}, we show how to linearly reduce any general case into a vertex-constrained case. 

\section{Preliminary}\label{sec:pre}

In the following, unless otherwise stated, we assume that our problem is the vertex-constrained case where every location of $\calP$ is at a vertex on $G$ and every vertex holds at least one location of $\calP$. Note that Lemma~\ref{lem:reduction} shows that any general case can be reduced in linear time into a vertex-constrained case. 

Some terminologies are borrowed from the literature~\cite{ref:BurkardAL98}. A $G$-vertex is a vertex on $G$ not included in any cycle, and a hinge is one on a cycle of degree greater than $2$. A graft is a maximum (connected) tree subgraph on $G$ where every leaf is either a hinge or a $G$-vertex, all hinges are at leaves, and no two hinges belong to the same cycle. A cactus graph is indeed a block graph consisting of graft blocks and cycle blocks so that the skeleton of $G$ is a tree $T$ where for each block on $G$, a node is joined by an edge to its hinges. See Fig.~\ref{fig:skeleton} for an example. 

\begin{figure}[ht]
\centering
\begin{minipage}{0.48\linewidth}
 \centering
 \includegraphics[width=0.95\textwidth]{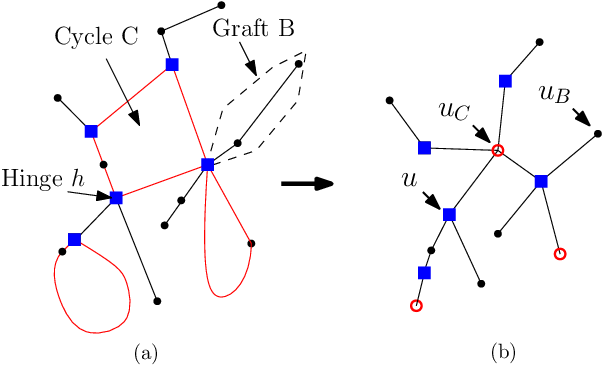}
 \caption{(a) Illustrating a cactus $G$ that consists of $3$ cycles, $5$ hinges (squares) and $6$ G-vertices (disks); (b) Illustrating $G$'s skeleton $T$ where circular and disk nodes represent cycles and grafts of $G$ (e.g., nodes $u$, $u_C$ and $u_B$ respectively representing hinge $h$, cycle $C$ and graft $B$ on $G$).}
 \label{fig:skeleton}
\end{minipage}
\hfill
\begin{minipage}{0.48\linewidth}
 \centering
 \includegraphics[width=0.95\textwidth]{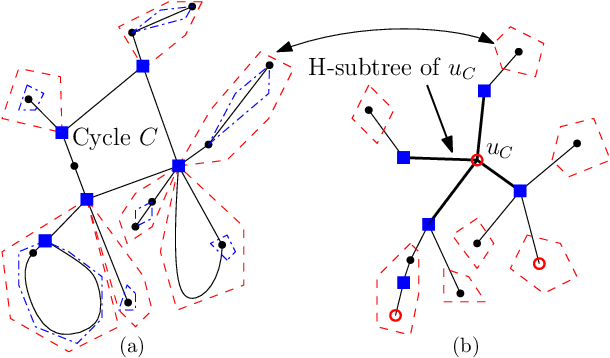}
 \caption{(a) Cycle $C$ on $G$ has $7$ split subgraphs (blue dash doted lines) and accordingly $7$ hanging subgraphs (red dashed lines); (b) on $T$, the H-subtree of node $u_c$ representing cycle $C$ has $7$ split subtrees each of which represents a distinct hanging subgraph of $C$ on $G$.}
\label{fig:hangingsub}
\end{minipage}
\end{figure}

In fact, $T$ represents a decomposition of $G$ so that we can traverse nodes on $T$ in a specific order to traverse $G$ blocks by blocks in the according order. Our algorithm thus works on $T$ to compute center $x^*$. Tree $T$ can be computed by a depth-first-search on $G$~\cite{ref:BurkardAL98,ref:Ben-MosheEf07} so that each node on $T$ is attached with a block or a hinge of $G$. We say that a node $u$ on $T$ is a block (resp., hinge) node if it represents a block (resp., hinge) on $G$. In our preprocessing work, we construct the skeleton $T$ with additional information maintained for nodes of $T$ to fasten the computation. 

Denote by $G_u$ the block (resp., hinge) on $G$ of any block (resp., hinge) node $u$ on $T$. More specifically, we calculate and maintain the cycle circumstance for every cycle node on $T$. For any hinge node $h$ on $T$, $h$ is attached with hinge $G_h$ on $G$ (i.e., $h$ represents $G_h$). For each adjacent node $u$ of $h$, vertex $G_h$ also exists on block $G_u$ but with only adjacent vertices of $G_u$ (that is, there is a copy of $G_h$ on $G_u$ but with adjacent vertices only on $G_u$). We associate each adjacent node $u$ in the adjacent list of $h$ with vertex $G_h$ (the copy of $G_h$) on $G_u$, and also maintain the link from vertex $G_h$ on $G_u$ to node $h$. 

Clearly, the size $|T|$ of $T$ is $O(mn)$ due to $|G| = O(mn)$. It is not difficult to see that all preprocessing work can be done in $O(mn)$ time. As a result, the following operations can be done in constant time.

\begin{enumerate}
    \item Given any vertex $v$ on $G$, finding the node on $T$ whose block $v$ is on;
    \item Given any hinge node $h$ on $T$, finding vertex $G_h$ on the block of every adjacent node of $h$ on $T$; 
    \item Given any block node $u$ on $T$, for any hinge on $G_u$, finding the hinge node on $T$ representing it.   
\end{enumerate}

Consider every hinge on the block of every block node on $T$ as an {\em open} vertex that does not contain any locations of $\calP$. To be convenient, for any point $x$ on $G$, we say that a node $u$ on $T$ contains $x$ or $x$ is on $u$ if $x$ is on $G_u$. Note that $x$ may be on multiple nodes if $x$ is at a hinge on $G$. We say that a subtree on $T$ contains $x$ if $x$ is on one of its nodes. 

Let $x$ be any point on $G$. Because $T$ defines a tree topology of blocks on $G$ so that vertices on $G$ can be traversed in some order. We consider computing $\Ed(P_i, x)$ for all $1\leq i\leq n$ by traversing $T$. We have the following lemma. Note that it defines an order of traversing $G$, which is used in other operations of our algorithm.   

\begin{lemma}\label{lem:computeEdforall}
Given any point $x$ on $G$, $\Ed(P_i, x)$ for all $1\leq i\leq n$ can be computed in $O(mn)$ time. 
\end{lemma}
\begin{proof}
We create an array $A[1\cdots n]$ to maintain all $\Ed(P_i,x)$ and initialize all as zero. Let $u_x$ be the block node on $T$ which contains $x$ and set it as the root of $T$. Clearly, $u_x$ as well as the corresponding point of $x$ on block $G_{u_x}$ can be obtained in $O(mn)$ time.

To compute $\Ed(P_i, x)$ for all $1\leq i\leq n$, it suffices to traverse $G$ starting from $x$ to compute the distance of every location to $x$. To do so, we instead traverse $T$ in the pre-order from $u_x$: During the traversal, block $G_{u_x}$ of $u_x$ is first traversed in the pre-order from $x$ to compute the distance of its every location to $x$. For every other block node $u$, the block is traversed in the pre-order starting from the hinge (open-vertex) whose corresponding hinge node on $T$ is the parent of $u$. So is every hinge node on $T$. 

More specifically, when we are visiting $G_{u_x}$, if $u_x$ is a cycle node then we traverse $G_{u_x}$ clockwise starting from $x$. During the traversal, for each vertex $v$, we first compute in constant time the distance $d(x,v)$; we next set $A[i] = A[i] + w_i\cdot f_{ij}\cdot d(x,v)$ for each location $p_{ij}$ at $v$ if $v$ is not a hinge; otherwise, we find in $O(1)$ time the hinge node $h$ on $T$ representing $v$ (i.e., $u_x$'s adjacent node), and set the distance $v$ on $h$ to $x$ as $d(x,v)$. 

In the case of $G_{u_x}$ being a graft, we perform the pre-order traversal from $x$ to update $A[1\cdots n]$ in the above way. Otherwise, $G_{u_x}$ is a hinge and so $d(G_{u_x},x) = 0$; we update $A[i]$ as the above for each location at $u_x$; we then set the distance $d(G_{u_x},x) = 0$ for $G_{u_x}$ on the block of every adjacent node of $u_x$. 

We continue our traversal on $T$ to visit $u_x$'s successors on $T$ in the pre-order to traverse their blocks. Suppose that we are visiting node $u$ on $T$. If $u$ is a hinge node, then its distance to $x$ can be known in constant time since hinge $G_u$ is an open vertex on the block of $u$'s parent node that has been visited. Consequently, we update $A[1\cdots n]$ as the above for every location $p_{ij}$ at $u$, and set the distance $d(G_u,x) = 0$ for $G_u$ on the block of every adjacent node of $u$.   

Otherwise, we traverse block $G_u$ from the hinge (open vertex) represented by $u$'s parent hinge node $h$ on $T$, which can be find in $O(1)$ time. As the distance of $G_{h}$ to $x$ has been known, the distance from every vertex on $G_u$ to $x$ can be obtained in $O(1)$ time. We thus update $A[1\cdots n]$ for locations on $G_u$ similarly. 

It follows that for any given point $x$ on $G$, values $\Ed(P_i, x)$ of all $1\leq i\leq n$ can be obtained in $O(mn)$ time by performing a pre-order traversal on $T$.\qed
\end{proof}

We say that a point $x$ on $G$ is an {\em articulation} point if $x$ is on a graft block; removing $x$ generates several connected disjoint subgraphs; each of them is called a {\em split} subgraph of $x$; the subgraph induced by $x$ and one of its split subgraphs is called a {\em hanging} subgraph of $x$. 

Similarly, any connected subgraph $G'$ of $G$ has several split subgraphs caused by removing $G'$, and each split subgraph with adjacent vertice(s) on $G'$ contributes a hanging subgraph. See Fig.~\ref{fig:hangingsub} (a) for an example. 

Consider any uncertain point $P_i\in\calP$. There exists a point $x^*_i$ on $G$ so that $\Ed(P_i,x)$ reaches its minimum at $x = x^*_i$; point $x^*_i$ is called the {\em median} of $P_i$ on $G$. For any subgraph $G'$ on $G$, we refer to value $\sum_{p_{ij}\in G'}f_{ij}$ as $P_i$'s {\em probability sum} of $G'$; we refer to value $w_i\cdot\sum_{p_{ij}\in G'}f_{ij}\cdot d(p_{ij},x)$ as $P_i$'s (weighted) {\em distance sum} of $G'$ to point $x$. 

Notice that we say that median $x^*_i$ of $P_i$ (resp., center $x^*$) is on a hanging subgraph of a subgraph $G'$ on $G$ iff $x^*_i$ (resp., $x^*$) is likely to be on that split subgraph of $G'$ it contains. We have the following lemma. 
 
\begin{lemma}\label{lem:medianatx}
Consider any articulation point $x$ on $G$ and any uncertain point $P_i\in\calP$.
    \begin{enumerate}
        \item If $x$ has a split subgraph whose probability sum of $P_i$ is greater than $0.5$, then its median $x_i^*$ is on the hanging subgraph including that split subgraph; 
        
        \item The point $x$ is $x_i^*$ if $P_i$'s probability sum of each split subgraph of $x$ is less than $0.5$;
        
        \item The point $x$ is $x_i^*$ if $x$ has a split subgraph with $P_i$'s probability sum equal to $0.5$.      
    \end{enumerate}
\end{lemma}
\begin{proof}
Let $G_1(x), \cdots, G_s(x)$ be all split subgraphs of $x$ on $G$. For claim $1$, we assume that $P_i$'s probability sum of $G_1(x)$ is larger than $0.5$. We shall show that $x^*_i$ is not likely to be on $G_k(x)$ for any $2\leq k\leq s$. 

Consider any split subgraph $G_k(x)$ with $2\leq k\leq s$. Let $y$ be any point on $G_k(x)$. By the expected distance definition, we have the following. 

\vspace{-0.15in}
\begin{center}
\begin{equation*}
\begin{split}
\Ed(P_i,y) 
& = w_i\sum_{p_{ij}\notin G_k(x)}f_{ij}[d(p_{ij},x)+d(x, y)]  + w_i\sum_{p_{ij}\in G_k(x)}f_{ij}d(p_{ij},y)  \\
& = w_i\sum_{p_{ij}\notin G_k(x)}f_{ij}d(p_{ij},x) + w_i\sum_{p_{ij}\notin G_k(x)}f_{ij}d(x, y)  + w_i\sum_{p_{ij}\in G_k(x)} f_{ij}d(p_{ij},y) \\
& > w_i\sum_{p_{ij}\notin G_k(x)}f_{ij}d(p_{ij},x) + w_i\sum_{p_{ij}\in G_k(x)} f_{ij}[d(x, y) + d(p_{ij},y)] \\
& > w_i\sum_{p_{ij}\notin G_k(x)}f_{ij}d(p_{ij},x) + w_i\sum_{p_{ij}\in G_k(x)} f_{ij}d(p_{ij},x) \\
& = \Ed(P_i, x)
\end{split}
\end{equation*} 
\end{center}

It follows that none of $G_2(x), \cdots, G_s(x)$ contain $x^*_i$ and $x^*_i$ is thus on the hanging subgraph $G_1(x)\cup\{x\}$. Therefore, both claims $1$ and $2$ hold. 

For claim $3$, suppose that $P_i$'s probability sum of $G_1(x)$ is equal to $0.5$. To prove claim $3$, it is sufficient to prove $\Ed(P_i,x)\leq\Ed(P_i,y)$ for any point $y\in G_1(x)$. This can be verified similarly and we thus omit the details.\qed
\end{proof}

For any point $x\in G$, we say that $P_i$ is a dominant uncertain point of $x$ if $\Ed(P_i,x)\geq\Ed(P_j,x)$ for each $1\leq j\leq n$. Point $x$ may have multiple dominant uncertain points. Lemma~\ref{lem:medianatx} implies the following corollary. 

\begin{corollary}\label{cor:centerforart}
Consider any articulation point $x$ on $G$. 
\begin{enumerate}
    \item If $x$ has one dominant uncertain point whose median is at $x$, then center $x^*$ is at $x$; 
    \item If two dominant uncertain points have their medians on different hanging subgraphs of $x$, then $x^*$ is at $x$;
    \item Otherwise, $x^*$ is on the hanging subgraph that contains all their medians. 
\end{enumerate}
\end{corollary}

Let $u$ be any block node on $T$; denote by $T^H_u$ the subtree on $T$ induced by $u$ and its adjacent (hinge) nodes; we refer to $T^H_u$ as the {\em H-subtree} of $u$ on $T$. Each hanging subgraph of block $G_u$ on $G$ is represented by a split subtree of $T^H_u$ on $T$. See Fig.~\ref{fig:hangingsub} (b) for an example. Lemma~\ref{lem:medianatx} also implies the following corollary. 

\begin{corollary}\label{cor:medianblock}
Consider any block node $u$ on $T$ and any uncertain point $P_i$ of $\calP$. 
\begin{enumerate}
    \item If the H-subtree $T^H_u$ of $u$ has a split subtree whose probability sum of $P_i$ is greater than $0.5$, then $x_i^*$ is on the split subtree of $T^H_u$; 
    
    \item If the probability sum of $P_i$ on each of $T^H_u$'s split subtree is less than $0.5$, then $x_i^*$ is on $u$ (i.e., block $G_u$ of $G$); 

    \item If $T^H_u$ has a split subtree whose probability sum of $P_i$ is equal to $0.5$, then $x_i^*$ is on that hinge node of $T^H_u$ that is adjacent to the split subtree. 
\end{enumerate}
\end{corollary}

Moreover, we have the following lemma.  

\begin{lemma}\label{lem:centerdetectart}
Given any articulation point $x$ on $G$, we can determine in $O(mn)$ time whether $x$ is $x^*$, and otherwise, which hanging subgraph of $x$ contains $x^*$. 
\end{lemma}
\begin{proof}
Apply Lemma~\ref{lem:computeEdforall} to compute the array $A[1\cdots n]$ with $A[i] = \Ed(P_i,x)$ in $O(mn)$ time, and then find the largest value $\delta$ of $A$ in $O(n)$ time. Create an array $F[1\cdots n]$ initialized as zero to store the probability sums of $x$'s dominant uncertain points on its each split subgraph, and another array $I[1\cdots n]$ initialized as $-1$ where $I[i]$ indicates $x$'s hanging subgraph containing $P_i$'s median $x_i^*$. 

We proceed with determining which hanging subgraph of $x$ contains medians of $x$'s dominant uncertain points by traversing $T$. 
Let $u_x$ be the node on $T$ containing $x$, which can be found in $O(mn)$ time. Notice that $u_x$ is either a hinge node or a graft node on $T$. Let $u_x$ be the root of $T$. 

On the one hand, $u_x$ is a graft node. Let $G^1_{u_x}, \cdots, G^s_{u_x}$ be the split subgraphs of $x$ on block $G_{u_x}$ of $u_x$. Hence, $x$ has $s$ split subgraphs $G_1(x), \cdots, G_s(x)$ on $G$ and $G^k_{u_x}\in G_k(x)$ for each $1\leq k\leq s$. Specifically, for each $1\leq k\leq s$, denote by $v^k_1, \cdots, v^k_t$ all hinges on $G^k_{u_x}$; since $G^1_{u_x}, \cdots, G^s_{u_x}$ are disjoint, the subgraph induced by $G_k(x)/G^k_{u_x}$ and $\{u^k_1, \cdots, u^k_t\}$ is represented by the union of subtrees on $T$ rooted at the corresponding hinge nodes $u^k_1, \cdots, u^k_t$ of $v^k_1, \cdots, v^k_t$. 

To prove the lemma, it suffices to compute the probability sum of dominant uncertain points of $x$ on each $G_k(x)$. For each $G_k(x)$, we maintain a list $L_k$ to store $u^k_1, \cdots, u^k_t$, which is empty initially. We then perform a traversal on $G_k(x)$ to compute the probability sum of uncertain points as follows. 

We first traverse $G^k_{u_x}$: For each non-hinge vertex $v$ on $G^k_{u_x}$, for each location $p_{ij}$ at $v$, we set $F[i] = F[i] +f_{ij}$; we then check whether $F[i]>0.5$ and $A[i] =\delta$; if both yes, then $P_i$ is a dominant uncertain point at $x$ whose median is on $G_k(x)\cup\{x\}$, and thereby we set $I[i]=k$; otherwise, $P_i$ is not a dominant uncertain point and hence we continue our traversal on $G^k_{u_x}$. When a hinge vertex $v$ is currently encountered, we find in $O(1)$ time its corresponding hinge node on $T$, add it to $L_k$, and then continue our traversal on $G^k_{u_x}$. 

Once we are done with traversing $G^k_{u_x}$, we continue to visit locations on the subgraph by $G_k(x)/G^k_{u_x}$ and $\{v^k_1, \cdots, v^k_t\}$. In order to do so, we traverse the subtree of $T$ rooted at each hinge node of $L_k$. The traversal is similar to that in Lemma~\ref{lem:computeEdforall} and so the details are omitted. 
    
Notice that after the above traversal on $G_k(x)$, we perform another traversal on $G_k(x)$ as the above, whereas during the traversal we reset $F[i]= 0$ for each location $p_{ij}$ on $G_k(x)$. Clearly, the traversal on all $G_k(x)$ can be carried out in $O(mn)$ time.  

To the end, we scan $I[1\cdots n]$ to determine which case of Corollary~\ref{cor:centerforart} $x$ falls into. More specifically, if there exist integers $i$ and $j$ with $1\leq i\neq j\leq n$ satisfying that $I[i], I[j] >0$ but $I[i]!=I[j]$, then two dominant uncertain points of $x$ have their medians on different hanging subgraphs of $x$ and so center $x^*$ must be at $x$; if $I[i] = -1$ for each $1\leq i\leq n$, $x^*$ is at $x$ as well; otherwise, only one hanging subgraph is found and it contains center $x^*$.

On the other hand, $u_x$ is a hinge node on $T$. Let $u_1, \cdots, u_s$ be all adjacent (block) nodes of $u_x$ on $T$. Denote by $T_{u_k}$ the subtree rooted at $u_k$ on $T$. Clearly, for each $1\leq k\leq s$, the subgraph represented by $T_{u_k}$ excluding vertex $G_{u_x}$ is exactly $G_k(x)$. Since $G_{u_x}$ is an open vertex on $G_{u_k}$, traversing each $G_k(x)$ amounts to traversing $T_{u_k}$, and we add only $u_k$ into $L_k$ for each $1\leq k\leq s$. It follows that we traverse $T_{u_k}$ to visit locations on $G_k(x)$ to compute $F[1\cdots n]$ and $I[1\cdots n]$ for each $1\leq k\leq s$; finally, we scan $I[1\cdots n]$ to determine as the above case where center $x^*$ locates. 

Therefore, the lemma holds.\qed
\end{proof}

Consider any hinge node $u$ on $T$. Lemma~\ref{lem:centerdetectart} implies the following corollary. 

\begin{corollary}\label{cor:centerdetectforhinge}
Given any hinge node $u$ on $T$, we can determine in $O(mn)$ time whether $x^*$ is on $u$ (i.e., at hinge $G_u$ on $G$), and otherwise, which split subtree of $u$ contains $x^*$. 
\end{corollary}

\section{The One-Center Problem on a Cycle}\label{sec:algforcycle}
In this section, we consider the one-center problem for the case of $G$ being a cycle. A general expected distance is considered: each $P_i\in\calP$ is associated with a constant $c_i$ so that the (weighted) distance of $P_i$ to $x$ is equal to their (weighted) expected distance plus $c_i$. With a little abuse of notations, we refer to it as the expected distance $\Ed(P_i,x)$ from $P_i$ to $x$. 

Our algorithm focuses on the vertex-constrained version where every location is at a vertex on $G$ and every vertex holds at least one location. Since $G$ is a cycle, it is easy to see that any general instance can be reduced in linear time to a vertex-constrained instance.  

Let $u_1, u_2, \cdots, u_M$ be the clockwise enumeration of all vertices on $G$, and $M\leq mn$. Let $l(G)$ be $G$'s circumstance. Every $u_i$ has a {\em semicircular} point $x_{i'}$ with $d(u_i,x_{i'}) = l(G)/2$ on $G$. Because sequence $x_{1'}, \cdots, x_{M'}$ is in the clockwise order. $x_{1'}, \cdots, x_{M'}$ can be computed in order in $O(mn)$ time by traversing $G$ clockwise. 

Join these semicircular points $x_{1'}, \cdots, x_{M'}$ to $G$ by merging them and $u_1, \cdots, u_M$ in clockwise order; simultaneously, reindex all vertices on $G$ clockwise.  Hence, a clockwise enumeration of all vertices on $G$ is generated in $O(mn)$ time. Clearly, the size $N$ of $G$ is now at most $2mn$. Given any vertex $u_i$ on $G$, 
there exists another vertex $u_{i^c}$ so that $d(u_i, u_{i^c}) = l(G)/2$. Importantly, $i^c = [{(i-1)}^c +1]\% N$ for $2\leq i\leq N$ and $1^c=(N^c +1)$. 

\begin{figure}[ht]
    \centering
    \includegraphics[width=0.45\textwidth]{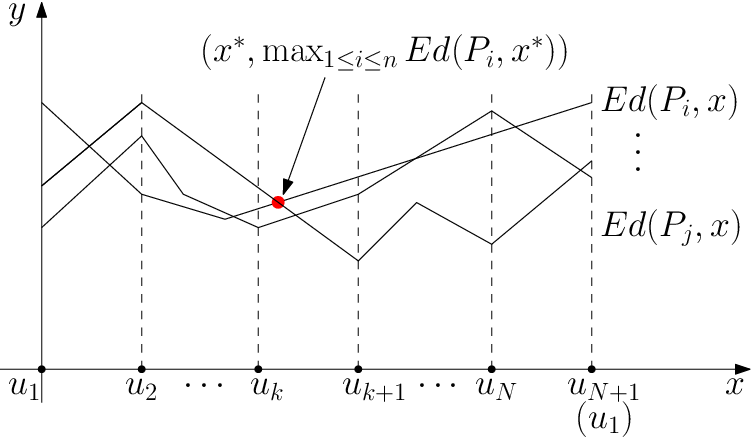}
    \caption{Consider $y=\Ed(P_i,x)$ in $x,y$-coordinate system by projecting cycle $G$ onto $x$-axis; $\Ed(P_i,x)$ of each $P_i\in\calP$ is linear in $x$ on any edge of $G$; center $x^*$ is decided by the projection on $x$-axis of the lowest point on the upper envelope of all $y=\Ed(P_i,x)$'s.}
    \label{fig:functions}
\end{figure}

Let $x$ be any point on $G$. Consider the expected distance $y=\Ed(P_i,x)$ in the $x,y$-coordinate system. We set $u_1$ at the origin and let vertices $u_1, u_2, \cdots, u_N$ $u_{N+1}, \cdots, u_{2N}$ be on $x$-axis in order so that $u_{N+i} = u_i$. Denote by $x_i$ the $x$-coordinate of $u_i$ on $x$-axis. For $1\leq i< j\leq N$, the clockwise distance between $u_i$ and $u_j$ on $G$ is exactly value $x_j-x_i$ and their counterclockwise distance is equal to $x_{i+N} - x_j$. 

As shall be analyzed below, each $\Ed(P_i,x)$ is linear in $x\in [x_s, x_{s+1}]$ for each $1\leq s\leq N$ but may be neither convex nor concave for $x\in [x_1, x_{N+1}]$, which is different to the deterministic case~\cite{ref:Ben-MosheEf07}. See Fig.~\ref{fig:functions}. Center $x^*$ is determined by the lowest point of the upper envelope of all $\Ed(P_i,x)$ for $x\in [x_1, x_{N+1}]$. Our strategy is computing the lowest point of the upper envelope on interval $[x_s, x_{s+1}]$, i.e., computing the local center $x^*_{s,s+1}$ of $\calP$ on $[x_s, x_{s+1}]$, for each $1\leq s\leq N$. Center $x^*$ is obviously decided by the lowest one among all of them. 

For each $1\leq s\leq N+1$, vertex $u_s$ has a {\em semicircular} point $x'$ on $x$-axis with $x_s-x'=l(G)/2$ and $x'$ must be at a vertex on $x$-axis in that $u_s$ on $G$ has its semicircular point at vertex $u_{s^c}$. We still let $u_{s^c}$ be $u_s$'s semicircular point on $x$-axis. Clearly, for each $1\leq s\leq N$, $(s+1)^c = s^c+1$, and the semicircular point of any point in $[x_s, x_{s+1}]$ lies in $[x_{s^c}, x_{(s+1)^c]}]$. Indices $1^c, 2^c, \cdots, (N+1)^c$ can be easily determined in order in $O(mn)$ time and so we omit the details. 

Consider any uncertain point $P_i$ of $\calP$. Because for any $1\leq s\leq N$, interval $[x_{s+1}, x_{s+N}]$ contains all locations of $\calP$ uniquely. We denote by $x_{ij}$ the $x$-coordinate of location $p_{ij}$ in $[x_{s+1}, x_{s+N}]$; denote by $F_i(x_s, x_{s^c})$ the probability sum of $P_i$'s locations in $[x_s, x_{s^c}]$; let $D_i(x_{s+1}, x_{s^c})$ be value $w_i\cdot\sum_{p_{ij}\in [x_{s+1}, x_{s^c}]} f_{ij}x_{ij}$ and $D^c_i(x_{s^c+1}, x_{s+N})$ be value $w_i\cdot\sum_{p_{ij}\in [x_{s^c+1}, x_{s+N}]} f_{ij}(l(G)-x_{ij})$. Due to $F_i(x_{s+1}, x_{s^c})+ F_i(x_{s^c+1}, x_{s+N}) =1$, we have that $\Ed(P_i, x)$ for $x\in [x_s, x_{s+1}]$ can be formulated as follows.  

\begin{center}
\begin{equation*}
    \begin{split}
        \Ed(P_i, x) & = c_i+ w_i\sum_{p_{ij}\in [x_{s+1}, x_{s^c}]} f_{ij}\cdot (x_{ij}-x) + w_i\sum_{p_{ij}\in [x_{s^c+1}, x_{s+N}]} f_{ij}\cdot [l(G)-(x_{ij}-x)] \\       
        & =c_i+ w_i(\sum_{p_{ij}\in [x_{s^c+1}, x_{s+N}]} f_{ij} - \sum_{p_{ij}\in [x_{s+1}, x_{s^c}]} f_{ij})\cdot x + w_i\sum_{p_{ij}\in [x_{s+1}, x_{s^c}]} f_{ij}x_{ij} \\ 
        & \ \ \ \ - w_i\sum_{p_{ij}\in [x_{s^c+1}, x_{s+N}]} f_{ij}(l(G)-x_{ij}) \\
        & =w_i [1-2F_i(x_{s+1}, x_{s^c})]\cdot x + c_i+  D_i(x_{s+1}, x_{s^c}) - D^c_i(x_{s^c+1}, x_{s+N})
        \end{split}
\end{equation*}
\end{center}

It turns out that each $\Ed(P_i,x)$ is linear in $x\in [x_s, x_{s+1}]$ for each $1\leq s\leq N$, and it turns at $x=x_s$ if $P_i$ has locations at points $u_s$, $u_{s^c}$, or $u_{s+N}$. Note that $\Ed(P_i,x)$ may be neither convex nor concave. Hence, each $\Ed(P_i,x)$ is a piece-wise linear function of complexity at most $m$ for $x\in [x_1, x_{N+1}]$. It follows that the local center $x^*_{s,s+1}$ of $\calP$ on $[x_s, x_{s+1}]$ is decided by the $x$-coordinate of the lowest point of the upper envelope on $[x_s, x_{s+1}]$ of functions $\Ed(P_i,x)$'s for all $1\leq i\leq n$. 

Consider the problem of computing the lowest points on the upper envelope of all $\Ed(P_i, x)$'s on interval $[x_s, x_{s+1}]$ for all $1\leq s\leq N$ from left to right. Let $L$ be the set of lines by extending all line segments on $\Ed(P_i,x)$ for all $1\leq i\leq n$, and $|L|\leq mn$. Since the upper envelope of lines is geometric dual to the convex (lower) hull of points, the dynamic convex-hull maintenance data structure of Brodal and Jacob~\cite{ref:BrodalDy02} can be applied to $L$ so that with $O(|L|\log |L|)$-time preprocessing and $O(|L|)$-space, our problem can be solved as follows. 

Suppose that we are about to process interval $[x_s, x_{s+1}]$. The dynamic convex-hull maintenance data structure $\Phi$ currently maintains the information of only $n$ lines caused by extending the line segment of each $\Ed(P_i,x)$'s on $[x_{s-1},x_{s}]$. Let $\calP_s$ be the subset of uncertain points of $\calP$ whose expected distance functions turn at $x=x_s$. For each $P_i\in\calP_s$, we delete from $\Phi$ function $\Ed(P_i, x)$ for $x\in [x_{s-1}, x_{s}]$ and then insert the line function of $\Ed(P_i, x)$ on $[x_s, x_{s+1}]$ into $\Phi$. After these $2|\calP_s|$ updates, we compute the local center $x^*_{s,s+1}$ of $\calP$ on $[x_{s}, x_{s+1}]$ as follows. 

Perform an extreme-point query on $\Phi$ in the vertical direction to compute the lowest point of the upper envelope of the $n$ lines. If the obtained point falls in $[x_s, x_{s+1}]$, $x^*_{s, s+1}$ is of same $x$-coordinate as this point and its $y$-coordinate is the objective value at $x^*_{s, s+1}$; otherwise, it is to left of line $x=x_s$ (resp., to right of $x=x_{s+1}$) and thereby $x^*_{s, s+1}$ is of $x$-coordinate equal to $x_s$ (resp., $x_{s+1}$). Accordingly, we then compute the objective value at $x=x_s$ (resp., $x=x_{s+1}$) by performing another extreme-point query in direction $y=-x_s\cdot x$ (resp., $y=-x_{s+1}\cdot x$). 

Note that $\calP_1 =\calP$ for interval $[x_1, x_2]$ and $\sum_{s=1}^{N}|P_s| = |L|\leq mn$. Since updates and queries each takes $O(\log |L|)$ amortized time, for each interval $[x_s, x_{s+1}]$, we spend totally $O(|\calP_s|\cdot\log |L|)$ amortized time on computing $x^*_{s,s+1}$. It implies that the time complexity for all updates and queries on $\Phi$ is $O(mn\log mn)$ time. Therefore, the total running time of computing the local centers of $\calP$ on $[x_s, x_{s+1}]$ for all $1\leq s\leq N$ is $O(mn\log mn)$ plus the time on determining functions $\Ed(P_i, x)$ of each $P_i\in\calP_s$ on $[x_s, x_{s+1}]$ for all $1\leq s\leq N$. 

We now present how to determine $\Ed(P_i,x)$ of each $P_i\in\calP_s$ in $x\in [x_s, x_{s+1}]$ for all $1\leq s\leq N$. Recall that $\Ed(P_i,x) = w_i\cdot[1-2F_i(x_{s+1}, x_{s^c})]\cdot x +D_i(x_{s+1}, x_{s^c}) - D^c_i(x_{s^c+1}, x_{s+N})+c_i$ for $x \in [x_s, x_{s+1}]$. It suffices to compute the three coefficients $F_i(x_{s+1}, x_{s^c})$, $D_i(x_{s+1}, x_{s^c})$ and $D^c_i(x_{s^c+1}, x_{s+N})$ for each $1\leq i\leq n$ and $1\leq s\leq N$. 

We create auxiliary arrays $X[1\cdots n]$, $Y[1\cdots n]$ and $Z[1\cdots n]$ to maintain the three coefficients of $\Ed(P_i, x)$ for $x$ in the current interval $[x_s, x_{s+1}]$, respectively; another array $I[1\cdots n]$ is also created so that $I[i] =1$ indicates that $P_i\in\calP_s$ for the current interval $[x_s, x_{s+1}]$; we associate with $u_s$ for each $1\leq s\leq N$ an empty list $\calA_s$ that will store the coefficients of $\Ed(P_i, x)$ on $[x_s, x_{s+1}]$ for each $P_i\in\calP_s$. Initially, $X[1\cdots n]$, $Y[1\cdots n]$, and $Z[1\cdots n]$ are all set as zero, and $I[1\cdots n]$ is set as one due to $\calP_1=\calP$.

For interval $[x_1, x_2]$, we compute $F_i(x_{2}, x_{1^c})$, $D_i(x_{2}, x_{1^c})$ and $D^c_i(x_{1^c+1}, x_{N+1})$ for each $1\leq i\leq n$: for every location $p_{ij}$ in $[x_{2}, x_{1^c}]$, we set $X[i] = X[i] + f_{ij}$ and $Y[i] = Y[i] + w_i\cdot f_{ij}\cdot x_{ij}$; for every location $p_{ij}$ in $[x_{1^c+1}, x_{N+1}]$, we set $Z[i] = Z[i]+ w_i\cdot (l(G)- x_{ij})$. Since $x_{1^c}$ is known in $O(1)$ time, it is easy to see that for all $P_i\in\calP_1$, functions $\Ed(P_i,x)$ for $x\in [x_1, x_2]$ can be determined in $O(mn)$ time. Next, we store in list $\calA_1$ the coefficients of $\Ed(P_i,x)$ of all $P_i\in\calP_1$ on $[x_1, x_2]$: for each $I[i]=1$, we add tuples $(i, w_i\cdot X[i], c_i+ Y[i]-Z[i])$ to $\calA_1$ and then set $I[i]=0$. Clearly, list $\calA_1$ for $u_1$ can be computed in $O(mn)$ time. 

Suppose we are about to determine the line function of $\Ed(P_i,x)$ on $[x_s, x_{s+1}]$, i.e., coefficients $F_i(x_{s+1}, x_{s^c})$, $D_i(x_{s+1}, x_{s^c})$ and $D^c_i(x_{s^c+1}, x_{s+N})$, for each $P_i\in\calP_s$. Note that if $P_i$ has no locations at $u_s$, $u_{s^c}$ and $u_{s+N}$, then $P_i$ is not in $\calP_s$; otherwise, $\Ed(P_i, x)$ turns at $x=x_s$ and we need to determine $\Ed(P_i, x)$ for $x\in [x_s,x_{s+1}]$. 

Recall that for $x\in [x_{s-1}, x_{s}]$, $\Ed(P_i,x) = c_i+ w_i\cdot [1-2F_i(x_{s}, x_{(s-1)^c})]\cdot x + D_i(x_{s}, x_{(s-1)^c}) - D^c_i(x_{(s-1)^c+1}, x_{s-1+N})$. On account of $s^c = (s-1)^c+1$, for $x\in [x_s, x_{s+1}] $, we have $F_i(x_{s+1}, x_{s^c}) = F_i(x_{s}, x_{(s-1)^c}) - F_i(x_s,x_s) + F_i(x_{s^c},x_{s^c})$, $D_i(x_{s+1}, x_{s^c}) =D_i(x_{s}, x_{(s-1)^c}) - D_i(x_{s}, x_{s}) + D_i(x_{s^c}, x_{s^c})$, and $D^c_i(x_{s^c+1}, x_{s+N}) = D^c_i(x_{(s-1)^c+1}, x_{s-1+N}) -D^c_i(x_{s^c}, x_{s^c}) + D^c_i(x_{s+N}, x_{s+N})$. Additionally, for each $1\leq i\leq n$, $\Ed(P_i,x)$ on $[x_{s-1}, x_{s}]$ is known, and its three coefficients are respectively in entries $X[i]$, $Y[i]$, and $Z[i]$. We can determine $\Ed(P_i,x)$ of each $P_i\in\calP_s$ on $[x_s,x_{s+1}]$ as follows. 

For each location $p_{ij}$ at $u_s$, we set $X[i] = X[i] - f_{ij}$, $Y[i] = Y[i] - w_i f_{ij} x_{ij}$ and $I[i] = 1$; for each location $p_{ij}$ at $u_{s^c}$, we set $X[i] = X[i] + f_{ij}$, $Y[i] = Y[i] + w_i f_{ij} x_{ij}$, $Z[i] = Z[i] - w_i f_{ij} (l(G) - x_{ij})$ and $I[i] = 1$; further, for each location $p_{ij}$ at $u_{s+N}$, we set $Z[i] = Z[i] + w_i f_{ij}(l(G) - x_{ij})$ and $I[i] = 1$. Subsequently, we revisit locations at $u_s$, $u_{s^c}$ and $u_{s+N}$: for each location $p_{ij}$, if $I[i] = 1$ then we add a tuple $(i, w_i\cdot X[i], c_i+ Y[i]-Z[i])$ to $\calA_s$ and set $I[i] =0$, and otherwise, we continue our visit. 

For each $2\leq s\leq N$, clearly, functions $\Ed(P_i,x)$ on $[x_s, x_{s+1}]$ of all $P_i\in\calP_s$ can be determined in the time linear to the number of locations at the three vertices $u_s$, $u_{s^c}$ and $u_{s+N}$. It follows that the time complexity for determining $\Ed(P_i,x)$ of each $P_i\in\calP_s$ for all $1\leq s\leq N$, i.e., computing the set $L$, is $O(mn)$; that is, the time complexity for determining $\Ed(P_i,x)$ for each $P_i\in\calP$ on $[x_1,x_{N+1}]$ is $O(mn)$. 

Combining all above efforts, we have the following theorem. 

\begin{theorem}\label{the:cycle}
The one-center problem of $\calP$ on a cycle can be solved in $O(|G|+mn\log mn)$ time. 
\end{theorem}

\section{The Algorithm}\label{sec:algorithm}
In this section, we shall present our algorithm for computing the center $x^*$ of $\calP$ on cactus $G$. We first give the lemma for solving the base case where a node of $T$, i.e., a block of $G$, is known to contain center $x^*$. 

\begin{lemma}\label{lem:basecase}
If a node $u$ on $T$ is known to contain center $x^*$, then $x^*$ can be computed in $O(mn\log mn)$ time. 
\end{lemma}
\begin{proof}
If $u$ is a hinge node, then $x^*$ is at its corresponding hinge $G_u$ on $G$, which can be obtained in $O(1)$ time, and we then return $G_u$ immediately. 

Otherwise, block $G_u$ of node $u$ is a graft or a cycle. Let $u$ be the root of $T$; let $u_1, \cdots, u_s$ be all child nodes of $u$, and each of them is a hinge node; vertices $G_{u_1}, \cdots, G_{u_s}$ are (open) vertices on $G_u$. Denote by $T_1(u), \cdots, T_s(u)$ the split subtrees of $u$ on $T$; for each $1\leq k\leq s$, $T_k(u)$ is rooted at $u_k$, and let $G_k(u)$ be the subgraph on $G$ that $T_k(u)$ represents. Note that $T_1(u), \cdots, T_s(u)$ can be known in $O(mn)$ time.  

On the one hand, $G_u$ is a graft and we then reduce our problem to an instance of the one-center problem with respect to a set $\calP'$ of $n$ uncertain points on a tree $G'$ so that center $x^*$ can be computed in $O(mn)$ time by the algorithm~\cite{ref:WangCo17} for tree graphs. 

Initialize $G'$ as $G_u$ and set $\calP'=\calP$. 
To reduce our problem to a tree instance, we then do a pre-order traversal on $T_k(u)$ from $u_k$ to traverse $G_k(u)$. More specifically, for hinge node $u_k$, we reassign all locations at $u_k$ to vertex $G_{u_k}$ on $G'$. For every other node $u'$ on $T_k(u)$, as in Lemma~\ref{lem:computeEdforall}, we traverse in the pre-order $G_{u'}$ from the hinge represented by its parent node: for each vertex $v$, we first compute distance $d(G_{u_k},v)$ and next check if $v$ is an open vertex. If no, we join a new vertex $v'$ into $G'$, set the edge length between $v'$ and $G_{u_k}$ on $G'$ as $d(G_{u_k},v)$, and reassign all locations of $v$ to $v'$; otherwise, we continue our traversal. 

Clearly, traversing all $T_k(u)$'s in the above way takes $O(mn)$ time in total. Now, we obtain a tree $G'$ of size $O(mn)$ and a set $\calP'$ of $n$ uncertain points where each $P_i\in\calP'$ has at most $m$ locations on $G'$. It is not difficult to see that the center of $\calP'$ on $G'$ corresponds a point on $G_u$ that is exactly the center of $\calP$ on $G$. Consequently, center $x^*$ can be computed in $O(mn)$ time by the algorithm~\cite{ref:WangCo17}.

On the other hand, $G_u$ is a cycle and we then reduce our problem into a cycle case where a set $\calP'$ of $n$ uncertain points are on cycle $G'$. Initially, we set $G'$ as $G_u$, set $\calP'$ as $\calP$, and assign a variable $c_i = 0$ to each $P_i\in\calP'$. Similarly, we do a pre-order traversal on each $T_k(u)$ from $u_k$ to traverse $G_k(u)$. For $u_k$, we reassign $G_{u_k}$'s locations to the copy of $G_{u_k}$ on $G'$. For every other node, we compute the distance $d(G_{u_k}, v)$ for each vertex $v$ of the block; if $v$ is not an open vertex, then we reassign each location $p_{ij}$ at $v$ to $G_{u_k}$ on $G'$, and set $c_i = c_i + w_if_{ij}\cdot d(G_{u_k}, v)$. 

The above $O(mn)$-time traversal generates a cycle $G'$ and a set $\calP'$ of $n$ uncertain points each with at most $m$ locations on $G'$ and a constant $c_i$. We can see that computing $x^*$ of $\calP$ on $G$ is equivalent to computing the center of $\calP'$ on $G'$, which can be solved in $O(mn\log mn)$ time by Theorem~\ref{the:cycle}. 

Hence, the lemma holds.\qed
\end{proof}

Now we are ready to present our algorithm that performs a recursive search on $T$ to locate the node, i.e., the block on $G$, that contains center $x^*$. Once the node is found, Lemma~\ref{lem:basecase} is then applied to find center $x^*$ in $O(mn\log mn)$ time. 

On the tree, a node is called a {\em centroid} if every split subtree of this node has no more than half nodes, and the centroid can be found in $O(|T|)$ time by a traversal on the tree~\cite{ref:KarivAn79,ref:MegiddoLi83}.

We first compute the centroid $c$ of $T$ in $O(|T|)$ time. If $c$ is a hinge node, then we apply Corollary~\ref{cor:centerdetectforhinge} to $c$, which takes $O(mn)$ time. If $x^*$ is on $c$, we then immediately return its hinge $G_c$ on $G$ as $x^*$; otherwise, we obtain a split subtree of $c$ on $T$ representing the hanging subgraph of hinge $G_c$ on $G$ that contains $x^*$. 

On the other hand, $c$ is a block node. We then solve the {\em center-detecting} problem for $c$ that is to decide which split subtree of $c$'s H-subtree $T^H_c$ on $T$ contains $x^*$, that is, determine which hanging subgraph of block $G_c$ contains $x^*$. As we shall present in Section~\ref{sec:centerdetecting}, the center-detecting problem can be solved in $O(mn)$ time. It follows that $x^*$ is either on one of $T^H_c$'s split subtrees or $T^H_c$. In the later case, since $G_c$ is represented by $T^H_c$, we can apply Lemma~\ref{lem:basecase} to $c$ so that the center $x^*$ can be obtained in $O(mn\log mn)$ time. 

In general, we obtain a subtree $T'$ that contains center $x^*$. The size of $T'$ is no more than half of $T$. Further, we continue to perform the above procedure recursively on the obtained $T'$. Similarly, we compute the centroid $c$ of $T'$ in $O(|T'|)$ time; we then determine in $O(mn)$ time whether $x^*$ is on node $c$, and otherwise, find the subtree of $T'$ containing $x^*$ but of size at most $|T'|/2$.   

As analyzed above, each recursive step takes $O(mn)$ time. After $O(\log mn)$ recursive steps, we obtain one node on $T$ that is known to contain center $x^*$. At this moment, we apply Lemma~\ref{lem:basecase} to this node to compute $x^*$ in $O(mn\log mn)$ time. Therefore, the vertex-constrained one-center problem can be solved in $O(mn\log mn)$ time.  

Recall that in the general case, locations of $\calP$ could be anywhere on the given cactus graph rather than only at vertices. To solve the general one-center problem, we first reduce the given general instance to a vertex-constrained instance by Lemma~\ref{lem:reduction}, and then apply our above algorithm to compute the center. The proof for Lemma~\ref{lem:reduction} is presented in Section~\ref{sec:reduction}. 

\begin{lemma}\label{lem:reduction}
The general case of the one-center problem can be reduced to a vertex-constrained case in $O(|G| + mn)$ time. 
\end{lemma}

\begin{theorem}\label{the:onecenter}
The one-center problem of $n$ uncertain points on cactus graphs can be solved in $O(|G| + mn\log mn)$ time. 
\end{theorem}

\subsection{The Center-Detecting Problem}\label{sec:centerdetecting}
Given any block node $u$ on $T$, the center-detecting problem is to determine which split subtree of $u$'s H-subtree $T^H_u$ on $T$ contains $x^*$, i.e., which hanging subgraph of block $G_u$ contains $x^*$. If $G$ is a tree, this problem can be solved in $O(mn)$ time~\cite{ref:WangCo17}. Our problem is on cacti and a new approach is proposed below.

Let $G_1(u), \cdots, G_s(u)$ be all hanging subgraphs of block $G_u$ on $G$. For each $G_k(u)$, let $v_k$ be the hinge on $G_k(u)$ that connects its vertices with $G/G_k(u)$. $G_1(u), \cdots, G_s(u)$ are represented by split subtrees $T_1(u), \cdots, T_s(u)$ of $T^H_u$ on $T$, respectively. 

Let $u$ be the root of $T$. For each $1\leq k\leq s$, $T_k(u)$ is rooted at a block node $u_k$, and hinge $v_k$ is an (open) vertex on block $G_{u_k}$. Additionally, the parent node of $u_k$ on $T$ is the hinge node $h_k$ on $T^H_u$ that represents $v_k$. Note that $h_k$ might be $h_t$ for some $1\leq t\neq k\leq s$. For all $1\leq k\leq s$, $T_k(u)$, $h_k$, and $v_k$ on block $G_{u_k}$ can be obtained in $O(mn)$ time via traversing subtrees rooted at $h_1, \cdots, h_s$. 

For each $1\leq k\leq s$, there is a subset $\calP_k$ of uncertain points so that each $P_i\in\calP_k$ has its probability sum of $G_k(u)/\{v_k\}$, i.e., $T_k(u)$, greater than $0.5$. Clearly, $\calP_i\cap\calP_j = \emptyset$ holds for any $1\leq i\neq j\leq s$. 

Define $\tau(G_k(u)) = \max_{P_i\in\calP_k}\Ed(P_i,v_k)$. Let $\gamma$ be the largest value of $\tau(G_k(u))$'s for all $1\leq k\leq s$. We have the following observation. 

\begin{observation}\label{obs:centernotatgi}
If $\tau(G_k(u))<\gamma$, then center $x^*$ cannot be on $G_k(u)/\{v_k\}$; if $\tau(G_r(u)) = \tau(G_t(u)) = \gamma$ for some $1\leq r\neq t\leq s$, then center $x^*$ is on block $G_u$. 
\end{observation}
\begin{proof}
Let $G_k(u)$ be such hanging subgraph of $G_u$ with $\tau(G_k(u))<\gamma$. For each $1\leq r\neq k\leq s$, by Lemma~\ref{lem:medianatx}, every uncertain point $P_i\in\calP_r$ has $\Ed(P_i,x)\geq\Ed(P_i,v_r)$ for any point $x\in G_k(u)$. Additionally, $\tau(G_k(u))<\gamma$. Hence, the dominant uncertain point at $v_k$ can not belong to $\calP_k$. By Corollary~\ref{cor:centerforart}, center $x^*$ cannot be on $G_k(u)/\{v_k\}$. 

Suppose there are two subgraphs $G_r(u)$ and $G_t(u)$ with $\tau(G_r(u)) =\tau(G_t(u)) = \gamma$. To prove that $x^*$ is on $G_u$, it is sufficient to show that $x^*$ is on neither $G_r(u)/v_r$ nor $G_t(u)/v_t$. There are the two cases to consider. 

If $v_r\neq v_t$, every $P_i\in\calP_r$ has $\Ed(P_i, v_r)<\Ed(P_i, v_t)$ in that $P_i$'s probability sum of $G_r(u)$ is greater than $0.5$. Hence, the dominant uncertain point at $v_t$ cannot be in $\calP_t$, and likewise, the dominant uncertain point at $v_r$ is not in $\calP_r$. It implies that if $v_r\neq v_t$ then $x^*$ is on neither $G_r(u)/v_r$ nor $G_t(u)/v_t$. 

Otherwise, $v_r$ is indeed $v_t$. If the dominant uncertain points of $v_t$ are in neither $\calP_r$ nor $\calP_t$, then $x^*$ cannot be on $G_r(u)/\{v_r\}\cup G_t(u)/\{v_t\}$. Otherwise, the objective value at $v_t$ is $\gamma$ due to $\tau(G_r(u)) =\tau(G_t(u)) = \gamma$. Hence, there are at least two dominant uncertain points at $v_t$: one in $\calP_r$ determining $\tau(G_r(u))$ and the other in $\calP_t$ determining $\tau(G_t(u))$. By Corollary~\ref{cor:centerforart}, we have that $x^*$ is at $v_t$, namely, $x^*$ is on neither $G_r(u)/v_r$ nor $G_t(u)/v_t$. 

The observation thus holds. \qed
\end{proof}

Below, we first describe the approach for solving the center-detecting problem and then present how to compute values $\tau(G_k(u))$ for all $1\leq k\leq s$. 

First, we compute $\gamma=\max_{k=1}^{s}\tau(G_k(u))$ in $O(s)$ time. We then determine in $O(s)$ time if there exists only one subgraph $G_r(u)$ with $\tau(G_r(u)) =\gamma$. If yes, then center $x^*$ is on either $G_r(u)$ or $G_u$. Their only common vertex is $v_r$, and $v_r$ and its corresponding hinge $h_r$ on $T$ are known in $O(1)$ time. For this case, we further apply Corollary~\ref{cor:centerdetectforhinge} to $h_r$ on $T$. If $x^*$ is at $v_r$ then we immediately return hinge $v_r$ on $G$ as the center; otherwise, we obtain the subtree on $T$ that represents the one containing $x^*$ among $G_r(u)$ and $G_u$, and return it.  

On the other hand, there exist at least two subgraphs, e.g., $G_r(u)$ and $G_t(u)$, so that $\tau(G_r(u)) =\tau(G_t(u))= \gamma$ for $1\leq r\neq t\leq s$. By Observation~\ref{obs:centernotatgi}, $x^*$ is on $G_u$ and thereby node $u$ on $T$ is returned. Due to $s\leq mn$, we can see that all the above operations can be carried out in $O(mn)$ time. 

To solve the center-detecting problem, it remains to compute $\tau(G_k(u))$ for all $1\leq k\leq s$. We first consider the problem of computing the distance $d(v_k, x)$ for any given point $x$ and any given $v_k$ on $G$. We have the following result.

\begin{lemma}\label{lem:computingdistanceofukandx}
With $O(mn)$-time preprocessing work, given any hinge $v_k$ and any point $x$ on $G$, the distance $d(v_k,x)$ can be known in constant time.
\end{lemma}
\begin{proof}
For each $G_k(u)$ with $1\leq k\leq s$, as in Lemma~\ref{lem:computeEdforall}, we do a pre-order traversal on $T_k(u)$ starting from its root $u_k$ to calculate the distance $d(v_k, v)$ from every vertex $v$ on $G_k(u)$ to $v_k$, which can be done in $O(mn)$ time. Meanwhile, we associate every vertex $v$ on $G_k(u)/\{v_k\}$ with node $u_k$ on $T$ to indicate that $v$ uniquely belongs to $G_k(u)$. All these can be done in $O(mn)$ in total.

We proceed with traversing block $G_u$ to compute its inter-vertex distances for all vertices on $G_u$. If $u$ is a graft node, we pick any vertex on $G_u$ as its root $r$ and then preform a pre-order traversal on $G_u$ to compute the distance of each vertex to $r$. Further, we construct the lowest common ancestor data structure~\cite{ref:BenderTh00,ref:HarelFa84} on $G_u$ so that with $O(|G_u|)$ preprocessing time and space, the lowest common ancestor of any two vertices on $G_u$ can be obtained in constant time. 

Now, given are any two points $y$ and $z$ on $G_u$, and let $v_y$ (resp., $v_z$) be the closest vertex to $r$ that is adjacent to $y$ (resp., $z$). We first determine $v_y$ and $v_z$ in $O(1)$ time so that distances $d(y,r)$ and $d(z,r)$ can be known in $O(1)$ time. We then compute the lowest common ancestor $v'$ of $v_y$ and $v_z$ by performing a constant-time query on the data structure. Due to $d(y, z) = d(y, r) + d(z,r)-2d(v',r)$, 
$d(y,z)$ can be derived in constant time. 

Otherwise, $u$ is a cycle node. In this situation, starting from any vertex $r$, we traverse $G_u$ clockwise to compute the clockwise distance of every vertex to $r$. For any points $y$ and $z$ on $G_u$, $d(y,z)$, equal to the minimum of their clockwise and counterclockwise distances, can be obtained in $O(1)$ time. 

We now consider the problem of computing $d(v_k,x)$ for any given $v_k$ and point $x$ on $G$. Let $(v, v')$ be the edge that contains $x$ on $G$. Note that edge $(v, v')$ is either on $G_r(u)$ for some $1\leq r\leq s$ or on $G_u$. So, there are only three cases to consider. 

On the one hand, $v$ and $v'$ are associated with the same node $u_r$ on $T$. Recall that hinge node $h_r$ is adjacent to $u_r$ and $u$ on $T$. It represents hinge $v_r$ on $G_u$, and $v_r$ is an open vertex on block $G_{u_r}$. So, edge $(v, v')$ is on $G_r(u)/v_r$. We first obtain hinge $v_r$ on $G$ by $u_r$ in $O(1)$ time. If $v_r$ is exactly $v_k$, then $d(v_k, x)$ can be obtained in $O(1)$ time since $d(v,v_k)$ and $d(v',v_k)$ have been calculated ahead. Otherwise, hinges $v_r$ and $v_k$ are on block $G_u$. Since $d(v_r, x)$ and $d(v_r, v_k)$ are obtained in $O(1)$ time, $d(v_k,x)$, equal to their sum, can be known in constant time.   

If only one of $v$ and $v'$, say $v$, is associated with a node $u_r$ on $T$, then edge $(v, v')$ is on $G_r(u)$ and $v'$ is exactly hinge $v_r$ on $G_u$. Either $v_r$ is not $v_k$ (but both are on $G_u$), or $v_r=v_k$. For either case, distance $d(v_k,x)$ can be known in constant time. 

Otherwise, edge $(v, v')$ is on $G_u$, i.e., neither of $v$ and $v'$ are associated with any node on $T$. Clearly, distance $d(v_k,x)$ can be known in constant time. 

Therefore, the lemma holds. \qed 
\end{proof}

We now consider the problem of computing $\tau(G_k(u))$ for each $1\leq k\leq s$, which is solved as follows.  

First, we determine the subset $\calP_k$ for each $1\leq k\leq s$: Create auxiliary arrays $A[1\cdots n]$ initialized as zero and $B[1\cdots n]$ initialized as null. We do a pre-order traversal on $T_k(u)$ from node $u_k$ to compute the probability sum of each $P_i$ on $G_k(u)/{v_k}$. During the traversal, for each location $p_{ij}$, we add $f_{ij}$ to $A[i]$ and continue to check if $A[i]>0.5$. If yes, we set $B[i]$ as $u_k$, and otherwise, we continue our traversal on $T_k(u)$. Once we are done, we traverse $T_k(u)$ again to reset $A[i]=0$ for every location $p_{ij}$ on $T_k(u)$. Clearly, $B[i]=u_k$ iff $P_i\in\calP_k$. After traversing $T_1(u), \cdots, T_s(u)$ as the above, given any $1\leq i\leq n$, we can know to which subset $P_i$ belongs by accessing $B[i]$. 

To compute $\tau(G_k(u))$ for each $1\leq k\leq s$, it suffices to compute $\Ed(P_i, v_k)$ for each $P_i\in\calP_k$. In details, we first create an array $L[1\cdots n]$ to maintain values $\Ed(P_i, v_k)$ of each $P_i\in\calP_k$ for all $1\leq k\leq s$. We then traverse $G$ directly to compute values $\Ed(P_i, v_k)$. During the traversal on $G$, for each location $p_{ij}$, if $B[i]$ is $u_k$, then $P_i$ is in $\calP_k$. We continue to compute in constant time the distance $d(p_{ij}, v_k)$ by Lemma~\ref{lem:computingdistanceofukandx}, and then add value $w_i\cdot f_{ij}\cdot d(p_{ij}, v_k)$ to $L[i]$. It follows that in $O(mn)$ time we can compute values $\Ed(P_i, v_k)$ of each $P_i\in\calP_k$ for all $1\leq k\leq s$. 

With the above efforts, $\tau(G_k(u))$ for all $1\leq k\leq s$ can be computed by scanning $L[1\cdots n]$: Initialize each $\tau(G_k(u))$ as zero. For each $L[i]$, supposing $B[i]$ is $u_k$, we set $\tau(G_k(u))$ as the larger of $\tau(G_k(u))$ and $L[i]$. Otherwise, either $L[i]=0$ or $B[i]$ is null, and hence we continue our scan. These can be carried out in $O(n)$ time. 

In a summary, with $O(mn)$-preprocessing work, values $\tau(G_k(u))$ for all $1\leq k\leq s$ can be computed in $O(mn)$ time. Once values $\tau(G_k(u))$ are known, as the above stated, the center-detecting problem for any given block node $u$ on $T$ can be solved in $O(mn)$ time. The following lemma is thus proved. 

\begin{lemma}\label{lem:solvecenterdetecting}
Given any block node $u$ on $T$, the center-detecting problem can be solved in $O(mn)$ time. 
\end{lemma}

\section{Reducing the General Case to the Vertex-Constrained Case}\label{sec:reduction}
In this section, we present how to reduce the general case to a vertex-constrained case. In the following, we say that a vertex on $G$ is empty if there are no locations at the vertex. 

Let $C$ be a cycle on $G$ of only two hinges where all other vertices are empty. Denote by $\pi$ the shorter path on $C$ between two hinges. If the length of $\pi$ is $l(C)/2$, then let $\pi$ be any of their clockwise and counterclockwise paths on $C$. The following observation helps reduce the size of $G$. 

\begin{observation}\label{obs:cycletwohinges}
If center $x^*$ is on $C$, $x^*$ must be on $\pi$. 
\end{observation}
\begin{proof}
Since only two hinges are on $C$ and all other vertices are empty, every empty non-hinge vertex on $C$ can be removed from $C$. On purpose of analysis, we assume that $C$ contains only two hinges. 

Suppose that $\pi$ is the counterclockwise path between two hinges longer than their clockwise path. Join the semicircular point of every hinge as a new vertex to $C$. Let $\{u_1, u_2, u_3, u_4\}$ be their clockwise enumerations starting from hinge $u_1$. We thus have the following properties: $u_4$ is the other hinge; $u_2$ must be $u_4$'s semicircular point; $u_3$ must be that of $u_1$. 

Removing $C$ except for $u_1, u_4$ generates two disjoint subgraphs $G_1$ and $G_4$ where $u_1$ is on $G_1$ and $u_4$ is on $G_4$ (and which are not hanging subgraphs of $C$). All $mn$ locations of $\calP$ are on $G_1\cup G_4$. Denote by $\calP_1$ the subset of all uncertain points in $\calP$ each with its probability sum of $G_1$ at least $0.5$, and by $\calP_2$ the subset of uncertain points each with its probability sum of $G_4$ at least $0.5$. Hence, $\calP_1\cup\calP_2 =\calP$. 

Let $x$ be any point on $C$. Consider function $\Ed(P_i,x)$ of each $P_i\in\calP$ with respect to $x$. It is easy to see that each $\Ed(P_i,x)$ linearly increases as $x$ moves clockwise from $u_1$ to $u_2$ along edge $(u_1, u_2)$, and so does it as $x$ moves counterclockwise from $u_4$ to $u_3$ along edge $(u_4, u_3)$. This means that the objective value at any point of $(u_1,u_2)/\{u_1\}$ (resp., $(u_4,u_3)/\{u_4\}$) is larger than that at $u_1$ (resp., $u_4$). Thus, center $x^*$ is on neither $(u_1, u_2)/\{u_1\}$ nor $(u_4, u_3)/\{u_4\}$. 

What's more, for each $P_r\in\calP_1$, function $\Ed(P_r,x)$ monotonically increases from $\Ed(P_r, u_2)$ to $\Ed(P_r, u_3)$ as $x$ moves clockwise from $u_2$ to $u_3$ along edge $(u_2, u_3)$. It monotonically increases as well from $\Ed(P_r, u_1)$ to $\Ed(P_r, u_4)$ as $x$ moves counterclockwise from $u_1$ to $u_4$ on edge $(u_1, u_4)$, i.e., the path $\pi$. Importantly, the increasing rate (slope) of $\Ed(P_r,x)$ for $x$ on edge $(u_2, u_3)$ is same as that of it for $x$ on edge $(u_1, u_4)$. 

Consider function $y=\Ed(P_i,x)$ of each $P_i\in\calP$ for $x$ on both edges $(u_1,u_4)$ and $(u_2,u_3)$ in the $x,y$-coordinate system. Let the two edges be on $x$-axis with both $u_1$ and $u_2$ at the origin. For each $P_r\in\calP_1$, $\Ed(P_r,x)$ defines a line segment for $x\in [u_2, u_3]$ (resp., $x\in [u_1, u_4]$). The line segment of $\Ed(P_r,x)$ for $x\in [u_2, u_3]$ is parallel to that of $\Ed(P_r,x)$ for $x\in [u_1, u_4]$. Likewise, for each $P_t\in\calP_2$, the line segment of $\Ed(P_t,x)$ for $x\in [u_2, u_3]$ is parallel to that of $\Ed(P_t,x)$ for $x\in [u_1, u_4]$. 

The local center of $\calP$ on edge $(u_2,u_3)$ (resp., $(u_1,u_4)$) is decided by the lowest point on the upper envelope of line segments by $n$ functions $y=\Ed(P_i,x)$ for $x\in [u_2, u_3]$ (resp., for $x\in [u_1,u_4]$) on $x$-axis. 
Extending each line segment to a line. Because $\Ed(P_i, u_1)<\Ed(P_i, u_2)$ and $\Ed(P_i, u_4)<\Ed(P_i, u_3)$ for each $P_i\in\calP$. The upper envelope of functions $y=\Ed(P_i,x)$ for $x\in [u_2, u_3]$ is enclosed by that of functions $y=\Ed(P_i,x)$ for $x\in [u_1, u_4]$. It implies that the local center of $\calP$ on edge $(u_1, u_4)$ is of a smaller objective value than that of $\calP$ on edge $(u_2, u_3)$. Thus, center $x^*$ is not on edge $(u_2, u_3)$ either. 

Based on the above analysis, we have that center $x^*$ is not likely to be on the longer path between hinges $u_1$ and $u_4$ except for themselves. Therefore, center $x^*$ is on the shorter path $\pi$ of $u_1$ and $u_4$ on $C$. 

It is possible that the clockwise and counterclockwise paths between two hinges on $C$ are of same length. In this situation, $u_2$ must be at $u_1$, and $u_3$ must be at $u_4$. Because no locations of $\calP$ are on $C/\{u_1, u_4\}$. Every point on the clockwise path from $u_1$ to $u_4$ can be matched to a point their counterclockwise path in terms of the objective value, and vice versa. Recall that $\pi$ is either one of the two paths. The above implies that center $x^*$ is likely to be on $\pi$, and the other path can be removed from $G$.   

Therefore, the observation holds. \qed
\end{proof}

Now we consider the reduction from the general case where locations of $\calP$ can be anywhere on cactus $G$ to a vertex-constrained case on a set $\calP'$ of $n$ uncertain points and cactus $G'$ where all locations of $\calP'$ are at vertices of $G'$ and every vertex on $G'$ holds at least one location. 

At first, we perform a traversal on $G$ to join a new vertex to $G$ for every location interior of an edge on $G$. Recall that all locations on any edge $e$ of $G$ are given sorted. Hence, these can be done in $O(|G|+mn)$ time. At this point, we obtain a cactus $G_1$ whose size is at most $(|G|+mn)$ and every location of $\calP$ is at a vertex of $G_1$. 

Further, we perform a post-order traversal on $G_1$ to process cycles. For every cycle $C$, we first determine whether $C$ has only one hinge and all other vertices on $C$ are empty. If yes, then we remove $C$ from $G_1$ except for that hinge since center $x^*$ is not likely to be on $C$ except for that hinge. Otherwise, we check whether $C$ meets the condition that $C$ has only two hinges but no locations are on its non-hinge vertices. If yes, by Observation~\ref{obs:cycletwohinges}, the longer path of the two hinges on $C$ can be removed. For this situation, we perform another traversal on $C$ to compute the shortest path length $a$ of two hinges, remove $C$ except for two hinges, and finally connect the two hinges directly via an edge of length equal to $a$. Clearly, the above operations can be carried out in $O(|G|+mn)$ time and a cactus graph $G_2$ is generated. 

We proceed with performing another post-order traversal on $G_2$ to further reduce the graph size. During the traversal, we delete every empty vertex of degree $1$; for each empty vertex of degree $2$, we remove it from $G_2$ by merging its two incident edges. As a consequence, a cactus graph $G_3$ is obtained after the $O(|G|+mn)$-time traversal. 

At this moment, every cycle with at most two hinges consists of non-hinge vertices, and each of them holds locations of $\calP$. Every vertex of degree at most $2$ holds locations of $\calP$. Hence, every empty vertex on $G_3$ is of degree at least $3$. By these above properties, we have the following observation. 

\begin{observation}\label{obs:boundofemptyvertex}
There are no more than $3mn$ empty vertices on $G_3$. 
\end{observation}
\begin{proof}
Since every vertex of degree at most $2$ on $G_3$ is not empty, every empty vertex is either a $G$-vertex or a hinge. Denote by $X$ the number of empty vertices on $G_3$. $X$ is thus bounded by the number $X_G$ of empty $G$-vertices plus the number $X_H$ of hinges on $G_3$. 

For the purpose of analysis, we construct a tree $T'$ from $G_3$ as follows: For every cycle $C$ on $G_3$, we replace $C$ by a new vertex $v$, connect $v$ with $C$'s adjacent vertices (hinges) on $G_3$, and reassign locations of $\calP$ at $C$'s non-hinge vertices to $v$. Additionally, we remove empty hinges of degree $2$ by connecting its two adjacent vertices; note that the number of hinges we removed is no more than the number of cycles. Because every cycle on $G_3$ with at most two hinges must contain non-empty non-hinge vertex. On $T'$, every vertex of degree at most $2$ is not empty. Since there are at most $mn$ locations on $T’$, there are at most $mn$ vertices of degree at most $2$ on $T‘$. It means the number of vertices of degree at least $3$ is no more than $mn$. Thus, we have $X_G\leq mn$.


Moreover, the above analysis implies that the size of $T’$ is no more than $2mn$. Because the total number of hinges on $G_3$, i.e., $X_H$, is less than the total number of cycles and $G$-vertices. Thus, we have $X_H\leq 2mn$. 

Therefore, the observation holds. \qed
\end{proof}

Observation~\ref{obs:boundofemptyvertex} implies $|G_3|\leq 4mn$. Let $G'$ be $G_3$ and denote by $V'$ the set of empty vertices on $G'$. Initialize $\calP'$ as $\calP$. We below assign new locations for each $P_i\in\calP'$ to construct a vertex-constrained case on cactus $G'$ and $\calP'$. 

First, we compute $V'$ by traversing $G'$ in $O(mn)$ time. We then create new locations for every uncertain point of $\calP'$. Suppose we are about to process $P_i$ of $\calP'$. Pick any $3m$ (empty) vertices from $V'$; then create $3m$ additional locations each with the probability of zero for $P_i$; assign each of them to one of the $3m$ vertices; finally, remove these $3m$ vertices from $V'$. We perform the same operations for uncertain points of $\calP'$ until $V'$ is empty. Now, every vertex on $G'$ holds at least one location. Additionally, we obtain a set $\calP'$ of $n$ uncertain points where each uncertain point $P_i$ has at most $4m$ locations on $G'$, and its each location is at a vertex on $G'$. 

Clearly, with $O(|G| + mn)$-time construction, we obtain a vertex-constrained case for $\calP'$ on $G'$. It is not difficult to see that solving the general case on $G$ with respect to $\calP$ is equivalent to solving this vertex-constrained case on $G'$ with respect to $\calP'$, which can be solved by our algorithm in $O(mn\log mn)$ time.  

\section{Conclusion}
In this paper, we consider the (weighted) one-center problem of $n$ uncertain points on a cactus graph. It is more challenging than the deterministic case~\cite{ref:Ben-MosheEf07} and the uncertain tree version~\cite{ref:WangCo17} because of the nonconvexity and the $O(m)$ complexity of the expected distance function. We propose an $O(|G| + mn\log mn)$ algorithm for this problem, which matches the $O(|G| + n\log n)$ result for the deterministic case~\cite{ref:Ben-MosheEf07}. Our algorithm is a simple binary search on the skeleton $T$ of $G$ for the block of $G$ containing the center. To support the search, we, however, solve the center-detecting problem for any given tree subgraph or cycle on a cactus. Our solution generalizes the method proposed for this problem on a tree~\cite{ref:WangCo17} but still runs in linear time. Moreover, an $O(|G|+mn\log mn)$ approach for the one-center problem on a cycle is proposed. Our techniques are interesting in its own right and may find applications elsewhere.

\end{document}